\pdfoutput=1
\documentclass[journal]{IEEEtran}
\title{Aerial_Robots_Journal_Paper}
\IEEEoverridecommandlockouts
\usepackage[dvipsnames]{xcolor}
\usepackage{amsmath}
\usepackage{mathtools}
\usepackage{amsmath}
\usepackage{amssymb}
\usepackage{amsthm}
\usepackage{tikz}
\usepackage{nicefrac}
\usepackage{mathtools, cuted}
\usepackage{lipsum, color}
\usepackage{float}
\usepackage{afterpage}
\usepackage{graphicx}
\usepackage{epstopdf}
\usepackage{caption}
\usepackage{subfig}
\usepackage{tabularx}
\usepackage{cite}
\usepackage{balance}
\usepackage{textcomp}%
\usepackage{booktabs}
\usepackage{multirow}
\usepackage{accents}
\usetikzlibrary{patterns,snakes}

\mathtoolsset{showonlyrefs}

\renewcommand{\figurename}{Figure}

\DeclareMathOperator*{\Min}{min}
\newtheorem{theorem}{Theorem}
\newtheorem{corollary}{Corollary}
\newtheorem{lemma}{Lemma}
\newtheorem{proposition}{Proposition}
\newtheorem{assumption}{Assumption}

\begin{document}


\title{Joint Optimization of Transmission and Propulsion in UAV-Assisted Communication Networks}

\author{Omar J. Faqir, Eric C. Kerrigan, Deniz G\"und\"uz, and Yuanbo Nie
\thanks{\textcopyright 2019 IEEE.  Personal use of this material is permitted.  Permission from IEEE must be obtained for all other uses, in any current or future media, including reprinting/republishing this material for advertising or promotional purposes, creating new collective works, for resale or redistribution to servers or lists, or reuse of any copyrighted component of this work in other works.}
\thanks{The support of the EPSRC Centre for Doctoral Training in High Performance Embedded and Distributed Systems  (HiPEDS, Grant Reference EP/L016796/1) is gratefully acknowledged. D. G\"und\"uz acknowledges funding from the European Research Council through project BEACON (Grant no. 677854). }
\thanks{O. J. Faqir and Deniz G\"und\"uz are with the Department of Electrical \& Electronics Engineering, Imperial College London, SW7~2AZ, U.K. {\tt\small ojf12@ic.ac.uk}, {\tt\small d.gunduz@ic.ac.uk}}%
\thanks{Eric C. Kerrigan is with the Department of Electrical \& Electronic Engineering
and Department of Aeronautics, Imperial College London, London SW7~2AZ, U.K. {\tt\small e.kerrigan@imperial.ac.uk}}%
\thanks{Yuanbo Nie is with the Department of Aeronautics, Imperial College London, London SW7~2AZ, U.K. {\tt\small yuanbo.nie15@imperial.ac.uk}}%
}


\maketitle

\begin{abstract}
The communication energy in a wireless network of mobile autonomous agents should be defined to include the propulsion energy  as well as the transmission energy used to facilitate information transfer.
We therefore develop communication-theoretic and Newtonian dynamic models of the communication and locomotion expenditures of an unmanned aerial vehicle (UAV).
These models are used to formulate a novel nonlinear optimal control problem (OCP) for arbitrary networks of autonomous agents.
This is the first work to consider mobility as a decision variable in UAV networks with multiple access channels.
Where possible, we compare our results with known analytic solutions for particular single-hop network configurations.
The OCP is then applied to a multiple-node UAV network for which previous results cannot be readily extended.
Numerical results demonstrate increased network capacity and communication energy savings upwards of $70\%$ when compared to more na\"ive communication policies.
\end{abstract}
\IEEEpeerreviewmaketitle

\section{Introduction} \label{sec:Intro}
Unmanned aerial vehicles (UAVs) have diverse potential uses and are currently benefiting from cost reduction and increased on-board compute power. Energy consumption remains a limiting factor, with significant drain from transmission and propulsion energy.
In this work we derive a policy for joint control of mobility and transmission to minimize total communication energy in a network. We achieve this by formulating and solving a continuous time nonlinear optimal control problem (OCP). Importantly, we consider communication energy to be the sum of transmission and any propulsion energy used to facilitate communication i.e. when a UAV slows down to maintain access to favourable channels.
We develop a general dynamic transmission model based on physical layer communication-theoretic bounds of ergodic and outage capacities. This is combined with a Newtonian dynamics mobility model and possible network topology.

UAVs can cooperatively complete high-level network objectives, generally including tasks of data gathering/relaying and coordinating movement for the purpose of data gathering/relaying.
Data may be collected from the environment (e.g. target tracking, search and pursuit \cite{ko2002network}, mobile sensor networks \cite{wang2013robotic,thammawichai2018optimizing}) or from other nodes and infrastructure (e.g. using UAVs as supplementary network links \cite{zhan2011wireless}).
We determine energy-efficient strategies for performing this data gathering and aggregation in a mobile network. Until relatively recently most works regarding node mobility focused on mobile and vehicular ad-hoc networks (MANETs and VANETs respectively), where mobility is either random or largely determined by infrastructure \cite{bekmezci2013flying}.
Since neither MANETs or VANETs are fully autonomous, mobility is typically not a decision variable.

In \cite{zhang2017cellular}, optimal trajectories are designed for a cellular-enabled UAV to maintain UAV connectivity by formulating the problem as a sequence of cell-tower to UAV associations.
A reciprocal problem is addressed in \cite{mozaffari2017wireless}, where optimal transport theory is used to derive UAV to cell associations that minimize average network delay for an arbitrary geometry of ground users.
For uniformly distributed users, the signal-to-noise (SNR) based association is proven delay-optimal. By the same authors, \cite{mozaffari2016unmanned} constructs an analytic framework for rate analysis of terrestrial device-to-device communications overlaid with an interfering UAV network. These works largely neglect UAV mobility dynamics in problem formulations. In formulating our OCP below we will refer further to existing works relating to energy-efficient communication, or relevant transmission and mobility models.

In \cite{zeng2016throughput} a single UAV is used as a mobile relay between a stationary source and a sink. For fixed trajectories the throughput maximizing transmission scheme is found analytically, by a directional waterfilling from source to sink. For fixed transmission profile the problem is non-convex and an optimal trajectory is found through a sequence of convex optimizations.
By the same authors, \cite{zeng2016energy} develops a method to maximize the throughput per unit of communication energy of a single circular UAV loiter trajectory.
As part of an on-line control scheme, \cite{yan2014go} uses a linear program (LP) to decide how close a slow rolling-robot should get to its download link before transmitting in order to minimize energy expenditure.

The two user broadcast channel is characterized in \cite{wu2018capacity} for a UAV transmitting independent data to two isolated ground nodes. In particular the hover-fly-hover strategy is shown to be optimal. The trade off between a ground node's communication energy and UAV's propulsion energy is investigated in \cite{yang2018energy} for the particular case of circular or straight line flights. A pareto boundary is characterised in both cases. Maximizing the minimum throughput between a set of ground users and multiple UAV receivers  is investigated in \cite{wu2018joint}. The problem is initially formulated as a mixed integer nonconvex program. A relaxed nonconvex problem is proposed, and solved through a sequence of block coordinate descent -- iterating between trajectory and transmission optimizations -- where the former is solved through successive convex optimization.
AWGN communication-theoretic bounds are used in \cite{zeng2016throughput, lyu2016cyclical,zeng2016energy}. In \cite{zeng2016throughput} a single UAV is used as a mobile relay between a stationary source and a sink. For fixed trajectories the throughput-maximizing transmission scheme is obtained analytically by a directional water-filling from source to sink;  water-filling is a well-known power allocation scheme for parallel channels \cite[Chapter~5]{tse2005fundamentals} and is further discussed in Section~\ref{sec:singleNodeSim}. On the other hand, for a fixed transmission profile the problem is non-convex and a trajectory is determined iteratively through the solution of a sequence of convex optimizations.
The same authors also developed a method to maximize the throughput per unit of communication energy of a single loitering UAV flying  at a constant speed~\cite{zeng2016energy}. The above works consider restricted cases of the throughput-maximization problem. In the sequel the power minimization problem is addressed.

A predictive channel model accounting for indoor fading dynamics is developed for rolling robotic networks in \cite{yan2012robotic} and employed in \cite{ali2015optimal,ali2016motion}, but relies on a priori channel measurement. Furthermore, non-convexity is addressed by solving a sequence of appropriately defined convex optimization problems, whereas in this work we generate a control input from formulating a single nonlinear (possibly non-convex) OCP, but leaving it up to the solver as to how best to compute a solution. In our experience with state-of-the-art solvers, such as IPOPT, this can be  more efficient than defining a sequence of convex optimization problems a priori.

This work extends \cite{faqir2017joint,faqir2018predictive} by providing supporting analysis of special cases and extending simulation results.
In Section~\ref{sec:Prob}, we formulate the continuous-time OCP for joint optimization of transmission and mobility policies of an arbitrarily-sized network consisting of both static and mobile nodes.
The general OCP is non-convex, and will be solved numerically by nonlinear optimization solvers. However, in Section~\ref{sec:ConvexAnaly} we present a number of reformulations of the nonlinear constraints and cost that can make the problem easier to solve in practice, as well as a number of special cases under which we can assuredly solve the problem to global optima.
In Section~\ref{sec:SimResults} we analyse simple network configurations in order to gain new insights, and provide a comparison of our results to known solutions. A comparison of energy usage between our proposed scheme and other possible approaches is shown in Section~\ref{sec:RelayIoT}, before presenting a closed-loop simulation with channel state uncertainty.
Even in very simple topologies, savings of upwards of $70\%$ are shown to be possible.

\begin{figure}[t!]
	\centering
	\includegraphics[width=\columnwidth]{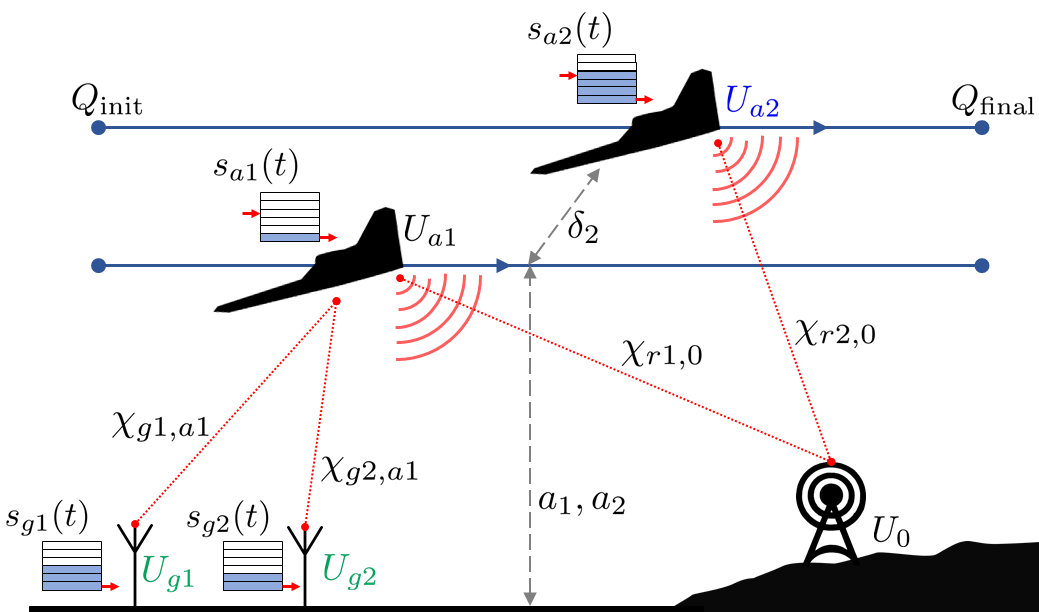}
	\caption[System Model Geometry]{Example system model and geometry for problem formulation. 
	Nodes with \textit{black} fonts will be presented in all our simulations.
	The node with a \textcolor{blue}{\textit{blue}} font title is part of the model studied in Section~\ref{sec:TwoNodeSim}, while the node with a \textcolor{ForestGreen}{\textit{green}} title is part of the model treated in Section~\ref{sec:RelayIoT}. The speeds of aerial nodes along these paths are variable and bounded.
	Solid lines represent the paths of UAV nodes, and red dashed lines correspond to existing communication links across distances $\chi$. Altitudes $a_{1},a_{2}=1$\,km, and displacement $\delta_{2}=1$\,km. For simplicity of exposition, we denote aerial nodes as $U_{ai}$ and ground nodes as $U_{gi}$, although they are modelled equivalently.}
	\label{fig:SystemModel}
\end{figure}

\section{System Model and Problem Formulation} \label{sec:Prob}
We consider a heterogeneous network of static and mobile nodes that collect/generate data and work cooperatively to aggregate this data at a specific subset of nodes, such as access points (AP) connected to wired infrastructure. 
 Figure~\ref{fig:SystemModel} exemplifies the simulation setup, with parameter definitions to follow.
Due to complexity issues, most UAV path planning algorithms restrict admissible trajectories to be constant-altitude and either linear or circular \cite{sujit2014unmanned}. For example, \cite{kang2009linear} uses nonlinear model predictive control (NMPC) for robust tracking of linear trajectories by fixed-wing UAVs.
For simplicity, we therefore consider $N$ (mobile) nodes $U_n, n \in \mathcal{N}\triangleq\{1,\hdots,N\}$ travelling along linear non-intersecting trajectories in a Cartesian space.
Denote the trajectory of $U_n$ over time interval $\mathcal{T} \triangleq [0,T]$ as $t \mapsto X_n(t) \triangleq (q_n(t), \delta_n, a_n)$, where $a_n$ and $\delta_n$ are the constant altitude and lateral displacement, and $q_n(t)$ the time-varying longitudinal displacement of $U_n$.
Over interval $\mathcal{T}$, each node $U_n$ must travel from position $q_n(0) = Q_{n,\text{init}}$ to $q_n(T)=Q_{n,\text{final}}$. Stationary terrestrial nodes are modelled with $Q_{n,\text{init}}=Q_{n ,\text{final}}$, $a_n=0$.
As in~\cite{sujit2014unmanned}, we define a trajectory as a time-parameterized path.

At time $t$, node $U_n$ stores data $s_n(t) \leq M_n$, where $M_n$ denotes the size of the node's on-board memory in bits. All storage buffers are subject to boundary conditions,
\begin{align}
s_n(0) = D_{n,\text{init}}, && s_n(T) \leq D_{n,\text{final}}, && \forall n \in \mathcal{N}.
\end{align}
Any node $U_n$ may be modelled as an ideal (infinite) sink with $D_{n,\text{final}} = M_n = \infty$. We may for example wish to model the existence of an infrastructure-connected AP in this way.

\subsection{Transmission Model}
Wireless communication links may exist from any node $U_n$ to $U_m, \forall n,m \in \mathcal{N}, n \not = m$, over channels with corresponding gains $h_{mn} \triangleq \nu_{mn}^2$, where $\nu_{mn}$ is a realization of the wireless channel gain. We define the link gain from $U_n$ to $U_m$ as
\begin{equation}
\eta_{mn}(\chi_{mn},h_{mn}) \triangleq \frac{h_{mn} G_{mn}}{\chi_{mn}^\alpha},
\end{equation}
where $\alpha>1$ is the path loss exponent, $G_{mn}\triangleq \tilde{G}_{mn}d_0^\alpha$ is a unitless constant of receive and transmit antenna gain $\tilde{G}_{mn}$ at reference distance $d_0$, and $\chi_{mn}$ is the squared distance between $U_n$ and $U_m$. We have
\begin{equation}
\chi_{mn}(t) \triangleq \|X_{mn}(t)\|^2 = \|(q_{mn}(t),\delta_{mn},a_{mn})\|^2,
\end{equation}
where $q_{mn}(t) \triangleq q_m(t)-q_n(t)$, and $\delta_{mn},a_{mn}$ are similarly defined.
At time $t$, node $U_n$ may transmit to node $U_m$ at a non-negative data rate $r_{mn}(t)$ using associated power $p_{mn}(t)$.

All nodes have a single omnidirectional antenna capable of a maximum transmission power of $P_\text{max}$ Watts.
We consider the case where an orthogonal frequency bandwidth $B_m$ is assigned for the reception of each node $U_m$. Each node may receive on its allocated bandwidth, while simultaneously transmitting on other bands. All messages destined for $U_m$ are transmitted over this band, forming a multiple access channel (MAC).
We do not allow coding (e.g.,\ network coding) or combining of different data packets at the nodes. Instead, we consider a decode-and-forward-based routing protocol at the relay nodes~\cite{gunduz2007opportunistic}.
The resulting network is a composition of MACs.
UAV communication links are typically dominated by line-of-sight (LoS) components, resulting in flat fading channels where all signal components undergo similar amplitude gains \cite{tse2005fundamentals}.
We consider the following channel modelling assumptions:

\subsubsection{Additive white Gaussian noise (AWGN)}
The channels between nodes are modelled as scalar AWGN with zero-mean, unit-variance, independent noise components. For a Gaussian MAC the set of achievable rate tuples defines a polymatroid capacity region resulting from the submodularity of mutual information \cite{tse1998multiaccess}.
If $N$ nodes transmit independent information to receiving terminal $U_m$ in the same communication interval, the received signal is a superposition of the $N$ transmitted signals scaled by their respective channel gains $\eta(\chi_{mn},h_n)$ plus an AWGN term.
The set of achievable data rates is evaluated using the Shannon capacity, which is an upper bound on achievable information rates subject to average power constraints. Any rate inside this capacity region may be transmitted with an arbitrarily small probability of error.

The capacity region $\mathcal{C}_{\tilde{N}}(\cdot)$ of a MAC formed by sources $U_n, n \in \tilde{\mathcal{N}} \subset\mathcal{N}$ and sink $U_m, m \in \mathcal{N} \setminus \tilde{\mathcal{N}}$ denotes the set of achievable rate tuples $r$, and is defined as
\begin{equation} \label{eq:CapRegionDef}
\mathcal{C}_{\tilde{N}}(\chi,p,h) \triangleq \left\{r \geq 0 \mid f_m(\chi,p,r,h,\mathcal{S}) \leq 0, \forall \mathcal{S} \subseteq \tilde{\mathcal{N}} \right\},
\end{equation}
where $\chi$ is the tuple of distances $\chi_{mn}$ between the $\tilde{N}$ users and $U_m$, $p \in \mathcal{P}^N$ is the $N$-tuple of transmission powers allocated by the $N$ users on the channel reserved for node $U_m$, with $\mathcal{P} \triangleq [0, P_{\text{max}}]$ being the range of possible transmission powers for each user.
$\mathcal{C}_{\tilde{N}}(\cdot)$ is bounded from above by $2^{\text{card}(\tilde{\mathcal{N}})}-1$ nonlinear submodular functions
\begin{multline} \label{eq:ShannonBound}
f_m(\chi,p,r,h,\mathcal{S}) \triangleq \\ \sum_{n \in \mathcal{S}} r_n - B_m \log_2 \left(1 + \sum_{n \in \mathcal{S}} \frac{\eta_{mn}(\chi_{mn},h_{mn})p_n}{\sigma_m^2} \right),
\end{multline}
where $r_n$ is the $n^\text{th}$ component of $r$, $\sigma_m^2 = 1$ is the receiver noise power and the channel gain $h_{mn}=1$ for AWGN channels.
Convexity of this region implies that throughput maximization does not require time-sharing between nodes, and can be achieved through the decoding process of successive interference cancellation (SIC) \cite{tse1998multiaccess}.
Since the channels are time-invariant, nodes are assumed to have perfect information regarding link status.

\subsubsection{Slow fading channel}
In a slow fading channel, the actual channel gains are random but remain constant over a certain communication interval, called the channel coherence time.
Considering \eqref{eq:CapRegionDef} with random vector $h$, we see that $\mathcal{C}_{\tilde{N}} = \emptyset$ with nonzero probability (assuming the transmitter has no channel state information; and hence, cannot perform power allocation).
Regardless of the transmission power and distance, it is impossible to guarantee successful transmission at any strictly positive rate with zero probability of error\footnote{This is under the assumption that $h_n$ cannot be bounded below by a positive value with probability $1$, that is, $P\{h_n\leq\epsilon\}>0,\forall \epsilon > 0$.}~\cite{tse2005fundamentals}.
As such it is no longer reasonable to model rates, power and distance using the capacity formulation in \eqref{eq:ShannonBound}.

Often the channel distribution is known or may be estimated, even if the actual channel state $h$ is unknown.
In this case we propose a more useful performance measure, the $\epsilon-$\emph{outage capacity} $\mathcal{C}_{\tilde{N}}^\epsilon$, defined as the set of achievable rates that guarantee a maximum outage probability of $\epsilon$, namely
\begin{equation}\label{eq:EpsCapRegionDef}
\mathcal{C}_{\tilde{N}}^\epsilon \triangleq \mathcal{C}_{\tilde{N}}(\chi,p,F_h^{-1}(1-\epsilon)),
\end{equation}
where $F_h$ is the complementary cumulative distribution of $h$, $F_h(x) \triangleq \text{Pr}\{h\geq x\}$ \cite{tse2005fundamentals}. In doing so we are performing chance-constrained optimization. However, because the probability density of $h$ is known, the problem may be written in a deterministic form with no additional complexity \cite{schwarm1999chance}.

\subsection{Propulsion Model}

In \cite{yan2014go, mei2004energy,tokekar2014energy}   the propulsion power required for a rolling-robot is modelled. respectively, as a linear function of speed, polynomial function of speed, and posynomial function of speed and acceleration. This posynomial model is further used in \cite{ali2015optimal,ali2016motion} for on-line communication and trajectory co-optimization. Where non-convexity is present, these works propose solving a series of successive convex problems, rather than the original problem.
We instead consider a fixed-wing UAV $U_n$, which is restricted to moving at positive speeds $v_n \in \mathcal{V}_n\triangleq[\underline{V}_n,\overline{V}_n],$ where $0 < \underline{V}_n \leq \overline{V}_n$. $D(\cdot)$ models the resistive forces on $U_n$, satisfying the following assumption:
\begin{assumption}\label{ass:umption1}
The resistive forces acting on node $U_n$ may be modeled by the function $v \mapsto D(v)$ such that $v \mapsto vD(v)$ is convex on the domain of admissible speeds $ v \in \mathcal{V}_n$ and $\infty$ on $v \not \in \mathcal{V}_n$.
\end{assumption}
The propulsion force $F_n(\cdot)$ generated by the UAV must satisfy the Newtonian dynamic force balance equation
\begin{equation} \label{eq:forceBalance}
F_n(t) - D(v_n(t)) = m_n a_n(t),
\end{equation}
where $m_n$ is the mass of the UAV $U_n$, $v_n(t)$ is its speed, and $a_n(t)\triangleq \dot{v}_n(t)$ is its acceleration along the direction of motion at time $t$.
The instantaneous power used for propulsion is the product $F_n(t)v_n(t)$, while the total propulsion energy is the integral of power over time \cite{zeng2016energy}.
For a fixed-wing UAV $v_n(t) \gg 0, \forall t \in \mathcal{T}$, whilst for a stationary terrestrial node $v_n(t) = \underline{V}_n=\overline{V}_n=0, \forall t \in \mathcal{T}$.

The drag force $D(v)$ of a fixed-wing UAV travelling at constant altitude and sub-sonic speed $v$ is modelled \cite{zeng2016throughput} as the sum of parasitic and lift-induced drag
\begin{equation} \label{eq:FullDragEq}
D(v) = \frac{\rho C_{D0}Sv^2}{2}   + \frac{2L^2}{(\pi e_0 A_R)\rho S v^2}.
\end{equation}
In \eqref{eq:FullDragEq}, parasitic drag is proportional to the square of the speed, where $\rho$ is air density, $C_{D0}$ is the zero lift drag coefficient, and~$S$ is the wing area. Lift-induced drag is inversely proportional to speed squared, where $e_0$ is the Oswald efficiency factor, $A_R$ the wing aspect ratio, and $L$ the lift force \cite{anderson2015introduction}. For level flight, $L$ equals the weight of the aircraft $W=mg$.

Motivated by \cite{zeng2016energy} and in agreement with Assumption~\ref{ass:umption1}, we model the resistive forces acting on the UAVs as
\begin{equation} \label{eq:NewtDragModel}
D(v) \triangleq \left\{
\begin{array}{ll}
C_{D1}v^2 + C_{D2}v^{-2}, \ &\forall v \in  \mathcal{V},\\
\infty,\ & \text{otherwise}
\end{array}
\right.
\end{equation}
where we have taken $C_{D1} = 9.26 \times 10^{-4}$ and $C_{D2} = 2250$ for our simulations, as in \cite{zeng2016energy}.

Although we specifically consider fixed-wing UAVs due to higher energy efficiency, rotor-craft may have practical advantages due to their ability to hover.
In \cite{mozaffari2017mobile} the energy used by a rotary craft moving at constant speed $v$ is decoupled as the sum of vertical and horizontal components.
Vertical power depends on the UAV mode of operation (climbing, descending, or descending in windmill state). Assumption~\ref{ass:umption1} is not satisfied in this case, since drag is not a smooth function of speed.

\subsection{Continuous-Time Optimal Control Problem Formulation}
Optimization is performed over the tuple of state and control variables which are denoted, for $U_n, n \in \mathcal{N}$, by $Y_n \triangleq (p_n,r_n,s_n,q_n,v_n,a_n,F_n),$
where $p_n$ is the tuple of outgoing transmission powers $p_{mn}(t), \forall m \in \mathcal{N}\setminus\{n\},$ and $r_n$ is the tuple of associated rates $r_{mn}$. The continuous-time OCP is
\mathtoolsset{showonlyrefs=false}
\begin{subequations} \label{eq:GeneralProbForm}
\begin{align}
& \Min_{Y_n,n \in \mathcal{N}}
\sum_{n=1}^{N} \int_{0}^{T} p_n(t) + v_n(t) F_n(t) \mathrm{d}t \label{eq:CostFunc} \\
   \text{s.t. } & \forall n,m \in \mathcal{N},  t \in \mathcal{T},\mathcal{S}\subseteq\mathcal{N} \notag
\end{align}
\begin{align}
& f_m(\chi(t),p(t),r(t),\tilde{h},\mathcal{S}\setminus\{m\}) \leq 0 \label{eq:CapRegionConst} \\
& \chi_{mn}(t) = \lVert X_{mn}(t) \rVert^2 \label{eq:CapRegionConst_Dis}\\
& \dot{s}_n(t) = \sum_{m \neq n} \left(r_{nm}(t)-r_{mn}(t)\right) \label{eq:storageUpdateConst}  \\
& s_n(0) = D_{n,\text{init}} , \quad s_n(T) = D_{n,\text{final}} \label{eq:initAndFinalStorageConst} \\
&F_n(t) - D(v_n(t)) = m_na_n(t) \label{eq:AccelerationBound} \\
&\dot{q}_n(t) = \Upsilon_{n} v_n(t) \label{eq:defVel1} \\
&\dot{v}_n(t) = a(t) \label{eq:defAcc1} \\
& q_n(0) = Q_{n,\text{init}}, \ q_n(T) = Q_{n,\text{final}} \label{eq:initAndFinalPosConst1} \\
& v_n(0) = v_{n,\text{init}}, \quad v_n(T) = v_{n,\text{final}} \label{eq:initVelocity} \\
& \underline{Y}_n \leq Y_n(t) \leq \overline{Y}_n    \label{eq:varBounds_orig}
\end{align}
\end{subequations}
\mathtoolsset{showonlyrefs=true}
The cost function \eqref{eq:CostFunc} is the sum of communication energy of all the nodes.
Dynamic stage constraints \eqref{eq:CapRegionConst}--\eqref{eq:CapRegionConst_Dis} bound the achievable data rates of each MAC to within the polymatroid capacity region of each receiving node.
$\tilde{h}=1$ for AWGN channels and $F^{-1}(1-\epsilon)$ for slow fading channels.
Stage constraint \eqref{eq:AccelerationBound} enforces the force balance condition.
System dynamics are included in \eqref{eq:storageUpdateConst}--\eqref{eq:defAcc1}, where \eqref{eq:storageUpdateConst} specifically updates data buffers with sent, received and collected data.
$\Upsilon_{n}\in\{-1,1\}$ depending on if position $q_n(t)$ decreases or increases respectively, because the speed $v_n(t)\geq 0$.

Boundary conditions \eqref{eq:initAndFinalStorageConst}--\eqref{eq:initVelocity} provide initial and final conditions on the state of the network. With reference to the discussion in Section~\ref{sec:Intro}, terminal constraints may be interpreted as the higher level objectives: by time $t=T$ all nodes must reach certain positions, and data must have been aggregated to certain nodes.  The simple bounds in~\eqref{eq:varBounds_orig} are given by
\begin{subequations}
\begin{align}
\underline{Y}_n &\triangleq (0,0,0,
-\infty,
\underline{V}_n,-\infty,\underline{F}), \\
\overline{Y}_n &\triangleq (P_{\text{max}},\infty,M,
\infty,
\overline{V}_n,\infty,\overline{F}),
\end{align}
\end{subequations}
where $0 \leq \underline{V}_n \leq \overline{V}_n$ and $\underline{F} \leq \overline{F}$.
The OCP can be discretized and solved using optimal control software, e.g.\ {ICLOCS}~\cite{ICLOCS2}.

Since no explicit routing is performed, the number of capacity region constraints is combinatorial in $\mathcal{N}$.
However, the complexity is not exponential in the absolute size of the network, but in the subset of nodes transmitting on a single MAC.
Therefore, our results are equally well suited to small networks or large networks with structure and/or partitioning.
Partitioning often arises due to the finite transmission range of the nodes, particularly in dense environments \cite{mozaffari2017mobile}.
A predefined hierarchical structure, such as the tree network used in \cite{nazemi2016qoi} also results in a highly structured network.

\section{Problem Analysis} \label{sec:ConvexAnaly}
The general problem \eqref{eq:GeneralProbForm} is non-convex.
We present a number of reformulations of the nonlinear constraints that can make the problem easier to solve in practice.
The nonlinear rate constraints \eqref{eq:CapRegionConst} are convex in transmission powers $p$, but are not jointly convex in both transmission powers and distances~$\chi$.
We will show that the nonlinear equality constraint~\eqref{eq:AccelerationBound} may be substituted into the cost function, convexifying the cost.
This, however, turns the previously simple thrust bound $F_\text{min} \leq F_n(t)$ into a concave constraint, unless thrust bounds are relaxed.
The absence of thrust bounds arises when considering a fixed trajectory, or is a reasonable assumption if the speed range is sufficiently small.
We finally give a number of special cases under which we may assure that all local optima are global optima.

\begin{lemma} \label{lemma:receivedSignal}
For a communication link from $U_n$ to $U_m$ the received signal strength, defined as
\begin{equation}
  \Gamma(p_n,\chi_{mn}) \triangleq \eta_{mn}(\chi_{mn},h_{mn})p_n = \frac{h_{mn}G_{mn}p_n}{\chi_{mn}^\alpha},
\end{equation}
is quasiconcave.
\end{lemma}
\begin{proof}
Take $G_{mn}h_{mn}=1$ for simplicity, and assume $\exists \pi_1\triangleq(x_1,y_1),\pi_2\triangleq(x_2,y_2) \in \mathbb{R}_+^2$ for which $\Gamma(\pi_i)\geq \beta, i\in\{1,2\}$. Now consider the point $\lambda \pi_1 + (1-\lambda)\pi_2, \lambda \in (0,1)$,
\begin{align}
\Gamma(\lambda \pi_1 + (1-\lambda) \pi_2) &= \frac{\lambda x_1 + (1-\lambda)x_2}{(\lambda y_1 + (1-\lambda)y_2)^\alpha} \\
& \geq \beta \frac{\lambda y_1^\alpha + (1-\lambda)y_2^\alpha}{(\lambda y_1 + (1-\lambda)y_2)^\alpha} \geq \beta,
\end{align}
which shows that the superlevel sets of $\Gamma(\cdot)$ are convex.
The first step follows from  $x_1 \geq \beta y_1^\alpha$ and $x_2 \geq \beta y_2^\alpha$. The second step follows from noting that $y \mapsto y^\alpha$ is convex on domain~$\mathbb{R}_+$, meaning that for $\lambda \in (0,1), \lambda y_1^\alpha + (1-\lambda)y_2^\alpha \geq (\lambda y_1 + (1-\lambda)y_2)^\alpha$, and hence the fraction is $\geq 1$.
\end{proof}

\begin{corollary} \label{corol:RateConstraint}
The rate constraints (\ref{eq:CapRegionConst}) are convex when considering transmission power optimization over a fixed trajectory, but are \emph{not} convex in the case of a free trajectory.
\end{corollary}
\begin{proof}
For receiver $U_m$ each capacity region constraint \eqref{eq:CapRegionConst} is of the form
\begin{equation} \label{eq:cor1Const}
\sum_{n \in S}r_{n}(t) - B_m\log_2 \left(
1 + \frac{G\tilde{h}_n}{\sigma^2}\sum_{n \in S} \frac{p_{n}(t)}{\chi_{mn}(t)^{\alpha}} \right) \leq 0.
\end{equation}
First, for a fixed trajectory \eqref{eq:cor1Const} is only a function of $r,p$, while $\chi$ is fixed.
The argument of the logarithm is linear in transmission powers. The function $\phi_1(x)\triangleq -\log(x)$ is convex, non-increasing. Since the composition of a convex, non-increasing function with a concave function is convex \cite{boyd2004convex}, and the linear combination of convex functions is also convex, \eqref{eq:cor1Const} is convex in $r,p$. See \cite{tse1998multiaccess} for further analysis.

When including the physical trajectory in the  optimization,~\eqref{eq:cor1Const} is a function of $r,p,q$.
The argument of the logarithm is now a sum of quasiconvex functions $\Gamma(\cdot)$ defined on separate domains.
A linear combination of quasiconvex functions is not quasiconvex, unless all functions but one are strictly convex~\cite{debreu1982additively}.
\end{proof}

\begin{lemma} \label{lemma:distance}
For the general problem \eqref{eq:GeneralProbForm} of minimizing communication energy, the relaxation of \eqref{eq:CapRegionConst_Dis} to the convex constraint
\begin{equation} \label{eq:relaxedDistance}
\chi_{mn}(t) \geq \lVert X_{mn}(t) \rVert^2
\end{equation}
does not change the solution.
\end{lemma}
\begin{proof}
Consider $Y^*$, the solution of \eqref{eq:GeneralProbForm}, with \eqref{eq:relaxedDistance} substituted instead of constraint \eqref{eq:CapRegionConst_Dis}. Assume $\exists t \in \mathcal{T} :\chi^*_{mn}(t) > ||X^*_{mn}(t)||^2$. If $p^*_{mn}(t)=0$ then the optimal cost is not dependent on $\chi^*_{mn}(t)$. Otherwise $p^*_{mn}(t)>0$ corresponds to a strictly positive rate $r^*_{mn}(t)$. Noting that rates are monotonically increasing in powers and monotonically decreasing in distances, the same rate $r^*_{mn}(t)$ may still be achieved with a power $\tilde{p}_{mn}(t) < p^*_{mn}(t)$ if the corresponding $\tilde{\chi}_{mn}(t) > ||X^*_{mn}(t)||^2$. Transmitting at power $\tilde{p}_{mn}(t)$ results in a strictly lower cost. Therefore $Y^*$ cannot be a minimizer of~\eqref{eq:GeneralProbForm}, with \eqref{eq:relaxedDistance} substituted instead of constraint \eqref{eq:CapRegionConst_Dis}, unless
\begin{equation}
\chi^*_{mn}(t) = ||X^*_{mn}(t)||^2, \forall t \in \mathcal{T}, \forall n,m \in \mathcal{N},\forall p^*_{mn}(t)>0.
\end{equation}
 This contradiction concludes the proof.
\end{proof}

Consider $\chi_{mn}$ as a slack variable representing the squared distance between nodes $U_m,U_n$. Apart from its definition~\eqref{eq:CapRegionConst_Dis}, it appears only in the data rate constraints \eqref{eq:CapRegionConst}, but not directly in the cost function or the dynamic constraints.

The posynomial objective function is also not convex over the whole of its domain and the logarithmic data rate term does not admit the use of geometric programming (GP) methods. However, convexification is possible by analysing the simplified problem in Lemma~\ref{lemma:force}.
\begin{lemma} \label{lemma:force}
\mathtoolsset{showonlyrefs=false}
The following problem
\begin{subequations}
\begin{align}
& \min_{v_n, F_n} \int_0^T F_n(t) v_n(t) \mathrm{d}t  \\
 \text{s.t.\ } & \forall t \in \mathcal{T}  \notag \\
& F_n(t) - D(v_n(t)) = m_n\dot{v}_n(t) \\
& \underline{F} \leq F_n(t) \leq \overline{F} \label{eq:vConst1} \\
& \underline{V} \leq v_n(t) \leq \overline{V} \label{eq:vConst2} \\
& v_n(0) = v_{n,\text{init}}, \quad v_n(T)=v_{n,\text{final}} \label{eq:vConst3}
\end{align}
\end{subequations}
\noindent of minimizing just the propulsion energy of a single node $U_n$ subject to thrust constraints, simple bounds, and initial and final conditions admits an equivalent convex form for all mappings $D$ satisfying Assumption~\ref{ass:umption1} and force bounds $(\underline{F},\overline{F})=(-\infty,\overline{F})$.
\mathtoolsset{showonlyrefs=true}
\end{lemma}
\begin{proof}
By noting that $F_n(t) = D(v_n(t)) + m_n\dot{v}_n(t)$, we move the equality into the cost function, rewriting the problem as
\begin{equation}
\min_{v_n} \phi(v_n)  \text{ s.t.\ \eqref{eq:vConst1}
--
\eqref{eq:vConst3}}, 
\end{equation}
where
\begin{equation}
\phi(v_n) \triangleq \underbrace{\int_0^T v_n(t)D(v_n(t)) \mathrm{d}t}_{\phi_1(v_n)} + \underbrace{m_n\int_0^T v_n(t)\dot{v}_n(t)\mathrm{d}t}_{\phi_2(v_n)}.
\end{equation}

We proceed by showing that both $\phi_1(\cdot)$ and $\phi_2(\cdot)$ are convex. Starting with the latter, by performing a change of variable, the analytic cost is derived by first noting that $\phi_2(v_n)$ is the change in kinetic energy
\begin{equation}
\phi_2(v_n) = m_n\int_{v_n(0)}^{v_n(T)} v dv = \frac{m_n}{2}\left(v_n^2(T)-v_n^2(0)\right),
\end{equation}
which is a convex function of $v_n(T)$ subject to fixed initial conditions \eqref{eq:vConst3}; in fact, it is possible to drop the $v_n^2(0)$ term completely without affecting the minimizing argument.
By Assumption~\ref{ass:umption1}, the mapping $D$ is convex and continuous. Since integrals preserve convexity, the total cost function $\phi(\cdot)$ is convex.

Having removed the thrust $F$ as a decision variable, satisfaction of input constraints would result in the set
\begin{equation}
	\mathcal{V}_F \triangleq \left\{ v_n \mid F_\text{min} \leq D(v_n(t)) + m_n\dot{v}_n(t) \leq F_\text{max} \right\}.
\end{equation}
Even with $D(\cdot)$ convex on the admissible range of speeds, the lower bound represents a concave constraint not admissible within a convex optimization framework. Dropping the lower bounds on thrust results in a final convex formulation:

\begin{subequations}
\begin{align}
& \min_{v_n} \int_0^T v_n(t)D(v_n(t)) \mathrm{d}t +  \frac{m_n}{2}\left(v_n^2(T)-v_n^2(0)\right) \\
 &\text{s.t. } \forall t \in \mathcal{T}  \notag\\
&\underline{V} \leq v_n \leq \overline{V} \\ 
&D(v_n(t)) + m_n\dot{v}_n(t) \leq \overline{F}\\
& v_n(0) = v_{n,\text{init}}, \quad v_n(T) = v_{n,\text{init}}.
\end{align}
\end{subequations}
\end{proof}

We now give two conditions for which all local solutions are global optima.
\begin{theorem}
For fixed trajectories, the problem \eqref{eq:GeneralProbForm}  is convex.
\end{theorem}
\begin{proof}
For fixed trajectories $\chi$, the decision variables are reduced to $(p_n, r_n, s_n),$ and constraints \eqref{eq:CapRegionConst_Dis}, \eqref{eq:AccelerationBound} and \eqref{eq:initVelocity})
may be omitted. The cost function is reduced to the sum of transmission powers. As such, convexity of the entire problem follows as a direct consequence of Corollary~\ref{corol:RateConstraint}.
\end{proof}

For the special case of a single UAV link with fixed trajectories this problem becomes one of power allocation over known time-varying channels, which has been addressed in various forms in the literature (e.g.\  \cite[Chapter~5]{tse2005fundamentals}). An example of this, with further analysis, is given in Section~\ref{sec:singleNodeSim}. In the multi-user setting for fading channels,  \cite{tse1998multiaccess} proposed solving this special case with a greedy algorithm for optimal rate and power allocation over MAC.

\begin{lemma} \label{lemma:QCComp}
Consider the monotonically non-increasing function $h:\mathbb{R}\rightarrow\mathbb{R}$ and the quasiconcave function $g:\mathcal{C}\subset\mathbb{R}^n\rightarrow \mathbb{R}$. The composition $h\circ g$ is quasiconvex.
\end{lemma}
\begin{proof}
From the quasiconcavity of $g$,
\begin{equation}
\min\{g(x),g(y)\} \leq g(\lambda x + (1-\lambda)y), \forall \lambda \in (0,1).
\end{equation}
Since $h$ is monotonically non-increasing, $a \leq b \Leftrightarrow h(a) \geq h(b)$, hence
\begin{equation}
h(\min\{g(x),g(y)\}) \geq h(g(\lambda x + (1-\lambda)y)).
\end{equation}
Noting that,
\begin{equation}
    h(\min\{g(x),g(y)\}) = \max\{h(g(x)),h(g(y))\},
\end{equation}
we may conclude
\begin{equation}
\max\{h(g(x)),h(g(y))\}\geq h(g(\lambda x + (1-\lambda)y)),
\end{equation}
and hence the quasiconvexity of $h\circ g$.
\end{proof}

In \cite[Chapter~3.4.4]{boyd2004convex} a similar statement as above is presented, but for composition of quasiconvex and non-decreasing functions.

\begin{theorem}
Consider a set of UAVs $U_n, n \in \mathcal{N}$. Assume that communication is restricted to multi-hop transmissions, where node $U_n$ transmits only to $U_{n+1}, \forall n \neq N$, reducing the communication network to single access channels. The problem of joint trajectory and transmission power optimization in order to sustain a constant minimum communication rate $\overline{r}$ is
\mathtoolsset{showonlyrefs=false}
\begin{subequations} \label{eq:singleUAVProb}
\begin{align}
& \Min_{Y_n, n \in \mathcal{N}} \sum_{n=1}^N \int_{0}^{T} p_n(t) + v_n(t) F_n(t) \mathrm{d}t \label{eq:CostFunc_1} \\
   \text{s.t. } & \forall n \in \mathcal{N}, \forall m \in \mathcal{N}\setminus\{N\}, \forall t \in \mathcal{T},\notag
\end{align}
\begin{align}
& -B_{(m+1)} \log_2\left(1 + \frac{G_{m(m+1)}}{\sigma^2} \left( \frac{p_m(t)}{\chi_{m(m+1)}(t)^\alpha}\right) \right) \leq - r, \label{eq:SingleUserConst_1}\\
& \chi_{m(m+1)}(t) = \|X_{m(m+1)}(t) \|^2, \label{eq:DistanceConst_1} \\
& q_n(0) = Q_{n,\text{init}}, \quad q_n(T) = Q_{n,\text{final}}, \label{eq:initAndFinalPosConst_1} \\
& v_n(0) = v_{n,\text{init}}, \label{eq:initVelocity_1} \\
&F_n(t)  = m_1\dot{v}_n(t) + \Omega(v_n(t)), \label{eq:AccelerationBound_1} \\
&\dot{q}_n(t) = \Upsilon_n v_n(t) \label{eq:defVel_1}. \\
& Y_{n,\text{min}} \leq Y_n(t) \leq Y_{n,\text{max}}.    \label{eq:varBounds_1}
\end{align}
\end{subequations}
In the absence of constraints on thrust, all local optima of the above problem are global optima.
\end{theorem}
\begin{proof}
Using Lemmas~\ref{lemma:distance} and \ref{lemma:force}, problem \eqref{eq:singleUAVProb} is equivalent to
\begin{subequations}
\begin{align*}
   &  \Min_{p,r,s,q,v} \sum_{n=1}^N\left[ \int_0^T p_n(t) + v_n(t)\Omega(v_n(t))\mathrm{d}t + \frac{m_n}{2} v_n^2(T) \right] \\
   \text{s.t. } & \forall n \in \mathcal{N}, \forall m \in \mathcal{N}\setminus\{N\}, \forall t \in \mathcal{T},\notag  \\
	& \text{\eqref{eq:SingleUserConst_1},\ \eqref{eq:relaxedDistance}, \eqref{eq:initAndFinalPosConst_1},\ \eqref{eq:initVelocity_1},\ \eqref{eq:defVel_1}},\   \tilde{Y}_{n,\text{min}} \leq \tilde{Y}_n(t) \leq \tilde{Y}_{n,\text{max}}
\end{align*}
\end{subequations}
where $\tilde{Y}_{n}(t) \triangleq \left(p_n(t),r_n(t),s_n(t),q_n(t),v_n(t)\right)$, and the bounds $\tilde{Y}_{n,\text{min}},$ and $ \tilde{Y}_{n,\text{max}}$ are similarly changed. The cost function and all the constraints, apart from \eqref{eq:SingleUserConst_1}, are convex. However, the function
\begin{equation}
    \xi(p_m,\chi_{m(m+1)}) \triangleq -\log \left(1 + \sigma^{-2} \Gamma(p_m,\chi_{m(m+1)}) \right)
\end{equation}
was shown to be quasiconvex. As noted in Lemma~\ref{lemma:receivedSignal}, $\Gamma(\cdot)$ is quasiconcave for all $\alpha>1$. From Lemma~\ref{lemma:QCComp}, the composition of a monotonically nonincreasing function with a quasiconcave function is quasiconvex. Therefore $\xi(\cdot)$ is quasiconvex. A direct implication is that the set
\begin{equation}
\mathcal{R}_m \triangleq\{(p_m,\chi_{m(m+1)}) | \xi(p_m,\chi_{m(m+1)}) \leq r\}
\end{equation}
is convex for each $m$. Since all the other constraints are linear, the constraint set is an intersection of half spaces with convex sets $\mathcal{R}_m$, which is convex.
Minimizing a convex cost over a convex set implies that all local optima are global optima.
\end{proof}

\section{Special Cases} \label{sec:SimResults}
The following analysis is for basic single-hop network topologies, with example settings depicted in blue in Figure~\ref{fig:SystemModel}.
The use of nonlinear models render the general problem~\eqref{eq:GeneralProbForm} non-convex, with non-trivial solutions.
Here we consider general AWGN channels. We also present particular cases of the problem, for which (known) analytic solutions exist. Our formulation allows  for new insights into these special cases. In Sections~\ref{sec:singleNodeSim} and \ref{sec:TwoNodeSim}, respectively, we focus on single and multiple access networks, with supporting numerical results presented in Section~\ref{sec:supNumres}.
Parameters used in all the simulations are defined in Table~\ref{tab:SimParameters}, where the chosen UAV speed range is consistent with \cite{bekmezci2013flying}.

\begin{table}[b!]
	\caption{Common model and simulation parameters.}
	\label{tab:SimParameters}
	\begin{tabular}{ | c | c | c | c | c | c | c | c |}
		\hline
		$\sigma^2$ & $B$ & $M$ & $P_{\max}$ & $\alpha$ & $T$ & $(\underline{V}_n,\overline{V}_n)$ & $m$\\
		\hline
		[W] & [Hz] & [GB] & [W] & [--] & [min] &  [m/s] & [kg]\\
		\hline
		$10^{-10}$ & $10^{5}$ & $1$ & $100$ & $1.5$ & $20$ & $(12,28)$ & 3\\
		\hline
	\end{tabular}
\end{table}

\subsection{Single UAV} \label{sec:singleNodeSim}
Consider a UAV $U_{a1}$ moving from $(Q_\text{init},0,a_{a1})$ to $(Q_\text{final},0,a_{a1})$, passing directly over a stationary AP $U_0$ positioned at $(0,0,0)$. Over time $T$, $U_{a1}$ is required to offload~$D_{a1}$ bits of data. We may simplify this problem by assuming the velocity profile of $U_{a1}$ is fixed and optimizing only over transmission policies. The predefined trajectory results in time-varying channel gains $\eta_{01}(t)$ which are fixed a priori. The optimal transmission scheme is then characterized by a water-filling solution \cite{wolf2011introduction}, which is a general term for equilization strategies used for power allocation in communication channels.
Water-filling allows us to cast the infinite-dimenisonal OCP as a single-dimensional problem.
A water-filling solution for \emph{rate maximization} may be found in \cite{wolf2011introduction}. In the following new result, we instead present a proof for \emph{power minimization}. Variable subscripts are dropped for notational simplicity.

\begin{proposition}\label{prop:water-filling}
	For a mobile transmitter with a predefined trajectory relative to a stationary receiver, over time interval~$\mathcal{T}$, the minimum transmission energy required to communicate~$D$ bits of data is found by solving
	\mathtoolsset{showonlyrefs=false}
	\begin{subequations} \label{eq:powerMin_prob}
		\begin{align}
		& p^\star \in \operatorname{arg} \Min_{p}  \int_0^T p(t) \mathrm{d}t && \\
		\text{s.t.} &
		 \int_{0}^{T}\ln\left(1 + \frac{\eta(t)p(t)}{\sigma^2}\right) \mathrm{d}t  = D, && \label{eq:singleNodeConstraint} \\
		& 0 \leq p(t) \leq P_\text{max}, && \forall t \in \mathcal{T},
		\end{align}
	\end{subequations}
	\mathtoolsset{showonlyrefs=true}
	which takes the form
	\begin{equation}
	p^\star(t) =  \min \left\{P_\text{max},\max\left\{0,\left(\zeta - \sigma^2/\eta(t)\right)\right\}\right\} \ ,
	\end{equation}
	where scalar $\zeta$ is a dual variable and $\eta(t)\triangleq \eta(\chi(t),h)$ is the time-varying channel gain due to fixed source trajectory.
\end{proposition}

\begin{proof}
	Isolate variables $p(t)$ by rewriting \eqref{eq:singleNodeConstraint} as
	\begin{equation}\label{eq:singleNodeConstraint2}
	\int_{0}^{T}\ln\left(\dfrac{\sigma^2}{\eta(t)} + p(t)\right) \mathrm{d}t  = \tilde{D},
	\end{equation}
	\begin{equation}
	\tilde{D} = D - \int_0^T \ln\left(\dfrac{\sigma^2}{\eta(t)}\right)\mathrm{d}t.
	\end{equation}
	The Lagrangian of \eqref{eq:powerMin_prob} is
	\begin{equation}	\label{eq:PM_Lagrangian}
	\begin{split}
	L(p,\zeta,\rho,\gamma) = & \int_0^T \big(- \rho(t) p(t) + \gamma(t)\left(p(t)-P_\text{max}\right)\big) \mathrm{d}t \\
	& - \zeta \left( \int_0^T \log_2\left(\dfrac{\sigma^2}{\eta(t)}+p(t)\right) \mathrm{d}t - \tilde{D} \right),
	\end{split}
	\end{equation}
	where dual variables $\rho(t),\gamma(t), \zeta$ correspond to the lower and upper bounds, and the integral data constraint.
	First-order optimality conditions result in the following solution,
	\begin{equation}\label{eq:powerMin_sol}
	p^\star(t) =
	\begin{cases}
	0,		&  \text{if }  \zeta \leq \dfrac{\sigma^2}{\eta(t)}, \\
	P_\text{max}, &  \text{if }  \zeta \geq  \left(\dfrac{\sigma^2}{\eta(t)} + P_\text{max}\right), \\
	\left(\zeta - \dfrac{\sigma^2}{\eta(t)} \right),  & \text{otherwise}. \\
	\end{cases}
	\end{equation}
\end{proof}

We may interpret $\sigma^2/\eta(t)$ as the effective noise power at time $t$ after normalizing with the channel gain. Intuitively, there exists a constant received power level $\sigma^2 + \eta(t)p(t)$ over $\mathcal{T}$ for which $D$ bits of data is communicated using minimal transmission energy. A binary search may be used to find $\zeta$.

The above result does not readily extend to when the source/receiver trajectory is not predetermined because the channel gains are no longer fixed. However, the transmission scheme of the jointly optimal solution to problem~\eqref{eq:GeneralProbForm} will be a water-filling solution of the channel gains corresponding to the optimal trajectory.
In some cases we may seek to determine the UAV trajectory $v(t)$ that maximizes data transfer, subject to peak power constraints and mobility dynamics. That is, we do not constrain the total energy consumption in order to characterise the maximum amount of data offloadable from the UAV.

\begin{proposition}\label{prop:rateMax}
	Consider a single-dimensional space, with a stationary receiver located at the origin and a mobile transmitter moving along a linear path from $0 < Q_\text{init} < Q_\text{final}$ over time $\mathcal{T}$. Without thrust constraints the data transfer $\int_\mathcal{T} r(t)\mathrm{d}t$ is maximized for the transmitter speed profile,
	\begin{equation}  \label{eq:optVelProf}
	v^*(t) =\left\{
	\begin{array}{l}
	\underline{V}, \ \forall t \in [0,t_1)\\
	\overline{V}, \ \forall t \in [t_1,T] .\\
	\end{array}
	\right.
	\end{equation}
\end{proposition}
\begin{proof}
	For maximum data transfer, we set $p(t) = P_\text{max},\forall t \in \mathcal{T}$. A trajectory is feasible if the node's speed satisfies the box constraints of set $\mathcal{V}$ and the node traverses the required distance, that is,
	\begin{equation}
	\int_0^T v(t) \mathrm{d}t = Q_\text{final}-Q_\text{init}, \quad v(t) \in [\underline{V}, \overline{V}], \forall t \in \mathcal{T}.
	\end{equation}
	For a feasible problem we have $\underline{V}T \leq Q_\text{final}-Q_\text{init} \leq \overline{V}T$. The single user capacity is a strictly decreasing function of the distance $\chi(\cdot)=q(\cdot)$ between transmitter and receiver. Furthermore, because $Q_\text{final}>Q_\text{init}$ and $\underline{V}>0$, $\chi(t)$ is strictly decreasing in $t$.
	Due to the monotonicity of the capacity function, a sufficient condition for optimality of the speed profile $v^*(\cdot)$, and correspondingly optimal distance $\chi^*(\cdot)$ is that
	\begin{equation} \label{eq:suffCondition}
	\chi^*(t) \leq \chi(t) = \int_0^t v(\tau)\mathrm{d}\tau, \forall t \in \mathcal{T},
	\end{equation}
	where $\chi(\cdot)$ is any continuous trajectory corresponding to a feasible speed $v(\cdot)$. Clearly, the position
	\begin{equation}
	\chi^*(t) = Q_\text{init} + \underline{V}t \leq \chi(t)
	\end{equation}
	satisfies \eqref{eq:suffCondition}, but is only feasible if it is still possible to reach the final destination by time $T$, i.e.
	\begin{equation}
	\chi^*(t) + \overline{V}(T-t) \geq Q_\text{final}.
	\end{equation}
	We use $t_1$ to denote the time when this is satisfied with equality. From position $\chi^*(t_1)$ at time $t_1$ the \emph{only} way to satisfy boundary conditions is to move at maximum speed $\overline{V}$ for remaining time $T-t_1$. Since $\chi^*(t_1) \leq \chi(t_1)$, and there exists only a single feasible speed profile that satisfies boundary conditions over $(t_1,T]$, the trajectory
	\begin{equation}
	\chi^*(t) = \left\{
	\begin{array}{ll}
	Q_\text{init} + \underline{V}t, & \text{if }t \in [0,t_1), \\
	Q_\text{init} + \underline{V}t_1 + \overline{V}(t-t_1), & \text{if }t \in [t_1,T],
	\end{array}
	\right.
	\end{equation}
	corresponding to \eqref{eq:optVelProf} must satisfy \eqref{eq:suffCondition} $\forall t \in \mathcal{T},$ where
	\begin{equation}
	t_1 \triangleq \left((Q_\text{final}-Q_\text{init})-\overline{V}T\right)(\underline{V}-\overline{V})^{-1}.
	\end{equation}
\end{proof}
For brevity we have only considered a UAV trajectory moving away from the source. By similar arguments we may see that in the case of $Q_\text{init}<0<Q_\text{final}$ (assuming non-zero UAV altitude and source located at origin) the optimal speed profile would be piecewise constant with
\begin{equation}  \label{eq:optVelProf2}
v^*(t) = \overline{V} \ \forall t \in [0,t_1)\cup (t_2,T]; \quad
\underline{V} \ \forall t \in [t_1,t_2] .
\end{equation}
If the UAV were able to hover (e.g.\ rotor-crafts) then the data-maximizing trajectory would require the UAV to hover at the point along its trajectory closest to the receiver -- analogous to the hover-fly-hover protocol \cite{wu2018capacity}.

\subsection{Two UAVs} \label{sec:TwoNodeSim}
Consider the transmission energy problem for UAVs $U_{a1},U_{a2}$ travelling along predefined trajectories (e.g.,\ the parallel trajectories shown in Figure~\ref{fig:SystemModel}) relative to stationary $U_0$. We allocate no bandwidth to $U_{a1},U_{a2}$ for receiving transmission. The result is a single MAC with $N=2$ transmitters $U_{a1},U_{a2}$ and receiver $U_0$.
The capacity region $\mathcal{C}_2(\chi,p,h)$ is the set of non-negative rate tuples $(r_{a1},r_{a2})$ satisfying
\mathtoolsset{showonlyrefs=false}
\begin{subequations}
	\label{eq:CapRegion2}
	\begin{align}
	0 \leq r_{a1}  &\leq \textstyle B_0\log_{2}\left(1+ \dfrac{\eta(\chi_{10},h_{a1})p_{a1}}{\sigma^2}\right) \label{eq:CapBound1}\\
	0 \leq r_{a2}  &\leq \textstyle B_0\log_{2}\left(1+ \dfrac{\eta(\chi_{20},h_{a2})p_{a2}}{\sigma^2}\right) \\
	r_{a1} + r_{a2} &\leq \textstyle B_0\log_{2}\left(1+ \dfrac{\eta(\chi_{10},h_{a1})p_{a1}+\eta(\chi_{20},h_{a2})p_{a2}}{\sigma^2}\right)  
	\end{align}
\end{subequations} 
for all $(p_{a1},p_{a2}) \in \mathcal{P}^2$. The first two are single-user bounds for each source. Information independence between $U_{a1},U_{a2}$ leads to the final constraint that the sum rate may not exceed the point-to-point capacity with full cooperation. For transmission powers $(p_{a1},p_{a2})$ the set of achievable rates is the pentagon in Figure~\ref{fig:CapacityRegion}.
\mathtoolsset{showonlyrefs=true} 
\begin{figure}[t!]
\centering
\vspace{0.5em}
\begin{tikzpicture}[scale=0.8]
	\draw[black, ultra thick, ->] (0,0) -- (0,5);
	\draw[black, ultra thick, ->] (0,0) -- (5,0);

	\draw[black,dashed] (0,4) -- (5,4);
	\draw[black,dashed] (4,0) -- (4,5);
	\draw[black,dashed] (1,5) -- (5,1);

	\draw[black, thick] (0,4) -- (2,4);
	\draw[black, thick] (4,0) -- (4,2);
	\draw[black, thick] (2,4) -- (4,2);

    \draw[black,fill=black,black] (4,2) circle (.4ex);
   	\draw[black,fill=black,black] (2,4) circle (.4ex);


	\node[draw=none, align=center] at (-0.5,4.8) {$r_{a1}$};
    \node[draw=none, align=center] at (4.8,-0.5) {$r_{a2}$};
    \node[draw=none, align=center] at (2.3,4.35) {$R^{(1)}$};
    \node[draw=none, align=center] at (4.5,2.3) {$R^{(2)}$};
    \node[draw=none, align=center] at (1,4.3) {$L_1$};
    \node[draw=none, align=center] at (4.35,1) {$L_2$};
    \node[draw=none, align=center] at (3.3,3.3) {$L_3$};

\end{tikzpicture}
    \caption[Two-user capacity region]{Capacity region for a given power policy across two parallel channels. Corner rates labelled as $R^{(1)}=(r_{a1}^{(1)},r_{a2}^{(1)})$ and $R^{(2)}=(r_{a1}^{(2)},r_{a2}^{(2)})$. Line segments labelled as $L_1, L_2, L_3$.}
    \label{fig:CapacityRegion}
\end{figure}
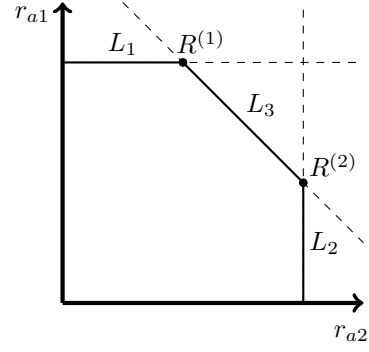
The rate tuple at vertex $R^{(1)}$ is achieved if the signal from $U_{a2}$ is decoded entirely before $U_{a1}$.
For a reversed decoding order the network operates at $R^{(2)}$.

With reference to Figure~\ref{fig:CapacityRegion}, the sum rate $r_{a1}+r_{a2}$ is maximized at any point on segment $L_3$. Therefore, for any given power tuple $(p_{a1},p_{a2})$, the optimal rate tuple will lie on segment $L_3$. This is formalized in Proposition~\ref{prop:SIC}. Equivalently we may construct any optimal rate pair $R^{(*)} = \left(r_{a1}^{(*)},r_{a2}^{(*)}\right)$ as the weighted sum
\begin{equation}
R^{(*)} = \varphi \cdot R^{(1)} + (1-\varphi) \cdot R^{(2)},
\end{equation}
for $\varphi \in [0,1]$.

The number of capacity region constraints grow exponentially with the number of MAC users. An important question is whether we can use the structure of the capacity region to simplify the problem statement. For the $N=2$ user case we observed that optimal rate points lie on the boundary $L_3$, which implies that the number of \textit{active} constraints at an optimal point scales at most linearly with the number of transmitters. This observation is formalized in the following lemma.

\begin{proposition} \label{prop:SIC}
    Consider a MAC with $N=2$ users $U_{a1},U_{a2}$ located at distances $\chi_{a1},\chi_{a2}$. For any arbitrary non-trivial rate pair $(r_{a1},r_{a2})$ the minimum power is achieved by first decoding the user with the better channel state, and subtracting this decoded signal from the remaining signal.
\end{proposition}
\begin{proof}
To emphasise that we are manipulating the power pairs $(p_{a1},p_{a2})$ to achieve a particular rate pair $R^{(*)},$ we rearrange~\eqref{eq:CapRegion2} to isolate transmission powers as
\begin{subequations}
	\label{eq:CapRegionProp}
	\begin{align}
	\sigma^2\left(2^{\frac{r_{a1}}{B_0}}-1\right)  &\leq \chi_{a1}^{-\alpha}p_{a1}, \label{eq:CapRegionProp1} \\
	\sigma^2\left(2^{\frac{r_{a12}}{B_0}}-1\right)  &\leq \chi_{a2}^{-\alpha}p_{a2}, \label{eq:CapRegionProp2} \\
	\underbrace{\sigma^2\left(2^{\frac{r_{a12} + r_{a2}}{B_0}}-1\right)}_{\Lambda}  &\leq \chi_{a1}^{-\alpha}p_{a1} + \chi_{a2}^{-\alpha}p_{a2}. \label{eq:CapRegionProp3}
	\end{align}
\end{subequations}
Say that AN arbitrary non-trivial rate pair $R^{(*)}=(r_{a1},r_{a2})$ is achieved with transmission power pair $(p_{a1},p_{a2})$. Due to the superlinearity of the $\exp{(\cdot)}$ function, \eqref{eq:CapRegionProp3} prevents both \eqref{eq:CapRegionProp1} and \eqref{eq:CapRegionProp2} from being simultaneously satisfied with equality.

Consider the case in which \eqref{eq:CapRegionProp3} holds with strict inequality. In this case, $(p_{a1},p_{a2})$ cannot be optimal in the sense of minimising $p_{a1}+p_{a2}$ because one or both of $(p_{a1},p_{a2})$ may be reduced while still satisfying \eqref{eq:CapRegionProp}. Therefore any optimal power allocation must satisfy \eqref{eq:CapRegionProp3} with equality. 

Now, given that \eqref{eq:CapRegionProp3} holds with equality, we rearrange to isolate $p_{a1},$ resulting in
\begin{equation} \label{eq:MACPowerAllocation}
    p_{a1} = \frac{\Lambda}{\chi_{a1}^{-\alpha}} - \left(\frac{\chi_{a1}}{\chi_{a2}}\right)^{\alpha}p_{a2}.
\end{equation}
If $\chi_{a1} > \chi_{a2}$, then the sum power may be reduced by increasing $p_{a2}$, while increasing $p_{a1}$ to satisfy \eqref{eq:MACPowerAllocation}, until constraint \eqref{eq:CapRegionProp1} holds with equality.
Otherwise, the sum power may be reduced by increasing $p_{a1}$ and reducing $p_{a2}$ until constraint \eqref{eq:CapRegionProp2} holds with equality. 
With reference to Figure~\ref{fig:CapacityRegion}, these  cases are equivalent to operating at $R^{(*)}=R^{(1)}$ or  $R^{(*)}=R^{(2)}$, respectively, achieved when $\varphi\in\{0,1\}$.
\end{proof}

Sum power-optimal decoding order leaves the user with the worst channel until last, independent of the data rates. Proposition~\ref{prop:SIC} shows that, for fixed trajectories, the set of active rate constraints may be determined offline.

\subsection{Numerical Results} \label{sec:supNumres}

Continuous-time problems are transcribed using {ICLOCS2}~\footnote{Transcription involves conversion of the original continuous time optimal control problem into a nonlinear program \cite{kelly2017introduction}. Various transcription methods exist, with the appropriate one often depending on characteristics of the problem. {ICLOCS2} supports  various different transcription methods. }~\cite{ICLOCS2} and numerically solved using the open source primal dual Interior Point solver {Ipopt} \cite{Wachter2006}.
{ICLOCS2} allows for rate constraints to be directly implemented on the discretized problem mesh. This prevents singular arcs and improves computational efficiency \cite{nie2018rate}. We use this feature to place derivative constraints \eqref{eq:defAcc1} on acceleration.
Energy usage for simulations discussed in this section may be found in Table~\ref{tab:BasicEnergyComp}.

\begin{table}[tb]
	\centering
	\caption{Transmission energy $\epsilon_T$ and propulsion energy $\epsilon_P$ usage of nodes $U_{a1},U_{a1}$ for simulations analysed in Section~\ref{sec:SimResults}. }
	\label{tab:BasicEnergyComp}
	\begin{tabular}{ | l | c | c | c | c |}
		\hline
		\multicolumn{1}{|c|}{Simulation}  & \multicolumn{2}{|c|}{$U_{a1}$ (kJ)} & \multicolumn{2}{|c|}{$U_{a2}$ (kJ)} \\
		\multirow{2}{*}{}  & $\epsilon_T$ & $\epsilon_P$ & $\epsilon_T$ & $\epsilon_P$ \\
		\hline
		$N=1$ fixed ($D_{a1}=45$\,MB) & 69.5  & 143.9 & --- & --- \\
		$N=1$ free  ($D_{a1}=65$\,MB) & 102.9 & 168.9 & ---  & ---\\
		$N=2$ fixed ($D_{a1},D_{a2}=22$\,MB)& 43.6 & 143.9 & 22.2 & 143.9 \\
		\hline
	\end{tabular}
\end{table}

We first present results for the single user case with fixed trajectory (Section~\ref{sec:singleNodeSim}). The solution is shown in Figure~\ref{fig:SingleNodeOptimalTransmitPower} for fixed UAV velocity profile
\begin{equation} \label{eq:constVelocity}
v_{a1}(t) = v_\text{avg} =\frac{1}{T} (Q_{a1,\text{final}}-Q_{a1,\text{final}}).
\end{equation}
Due to strict convexity of the drag function \eqref{eq:NewtDragModel}, this constant velocity profile uses minimum propulsion energy. Agreeing with Proposition~\ref{prop:water-filling}, an inverse relationship between $p_{a1}(t)$ and the effective noise $\sigma^2/\eta_{a1,0}(t)$ is shown, where $\zeta$ coincides with the peak transmission power. Here $U_{a1}$ was initialized with $D=50$\,MB. We observe that $U_{a1}$ transmits only within a certain proximity of the static destination node, and more power is allocated for transmission when it is closer to the destination.

\begin{figure}[t!]
	\centering
	\includegraphics[width=\columnwidth]{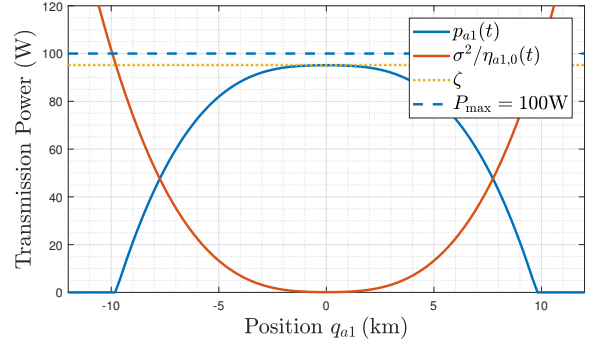}
	\caption[Single-node optimal transmit power scheme and associated data rate]{Optimal power and rate allocation for a single-node $U_{a1}$ at constant speed is characterized entirely by the scalar~$\zeta$.}
	\label{fig:SingleNodeOptimalTransmitPower}
\end{figure}
If instead we allow $U_{a1}$ to have a free trajectory then the jointly optimal transmission and mobility profiles of $U_{a1}$ generated by problem \eqref{eq:GeneralProbForm} are shown in Figure~\ref{fig:SingleNodeCommsProblemSimulationResults} for a greater starting load of $D=65$\,MB. $U_{a1}$ moves at velocity $\overline{V}$ when it is further from the transmitter, but then expends energy to slow down to $\underline{V}$ when it is close to the AP in order to maintain a better channel for a longer duration. During this time the UAV is transmitting at peak power in order to opportunistically exploit the favourable channel gain. Considering the insights of Lemma~\ref{prop:rateMax}, we may correctly surmise that $65$\,MB is close to the network capacity.

\begin{figure}[b!]
	\centering
	\vspace{0.5em}
	\subfloat[Optimal transmission power and thrust profile of $U_{a1}$.]{
		\label{subfig:singelNodeSpeedVarying1}
		\includegraphics[width=\columnwidth]{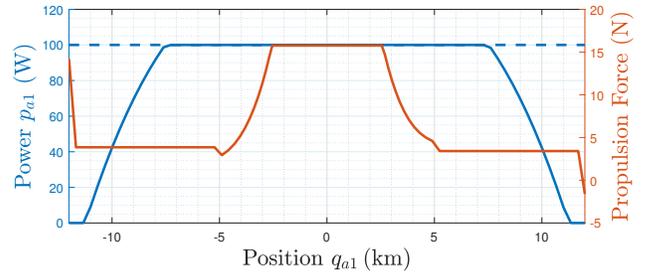}}

	\subfloat[Associated achieved data rate and velocity profile of $U_{a1}$.]{
		\label{subfig:singelNodeSpeedVarying2}
		\includegraphics[width=\columnwidth]{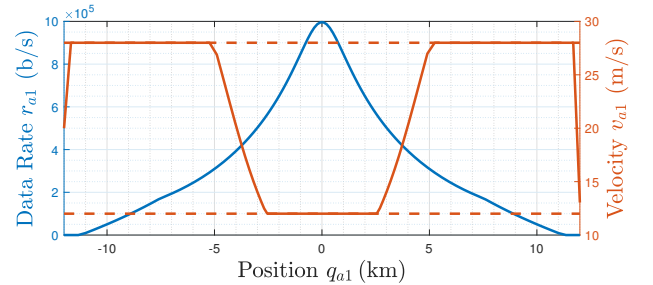}}

	\caption{Single-node problem. Dashed lines indicate bounds on respective variables corresponding to $\underline{Y}_n \leq Y_n(t) \leq \overline{Y}_n$.}
	\label{fig:SingleNodeCommsProblemSimulationResults}
\end{figure}

Simulation results for the two node fixed trajectory problem (Section~\ref{sec:TwoNodeSim}) are shown in Figure~\ref{fig:TwoNodeSimulation}, where $U_{a1},U_{a2}$ are initialized with data $D_{a1}=D_{a2}=22$\,MB and travel at fixed speeds of $72$\,km/h. The distances $\chi_{01}(t) < \chi_{02}(t), \forall t \in \mathcal{T}$ are such that
$U_{a1}$ experiences a more favourable channel at each time instance.
The transmission profile of $U_{a1}$ bears strong resemblance to the single user case. Interestingly, when both $U_{a1},U_{a2}$ are transmitting, $U_{a2}$ is able to increase transmission rate while decreasing transmission power. As shown in Table~\ref{tab:BasicEnergyComp}, this policy actually results in $U_{a2}$ using less transmission energy than $U_{a1}$. We may explain this with reference to Proposition~\ref{prop:SIC}.
The mapping $t\mapsto \varphi(t)$ may be a time-varying priority, and is only uniquely defined when $p_n(t)>0, n \in \{1,2\}$. If we calculate $\varphi(\cdot)$ from the optimal powers and rates shown in Figure~\ref{fig:TwoNodeSimulation}, we find that $\varphi(t)=0,\forall t \in \{t \in \mathcal{T}\mid p_n(t)>0,n \in \{1,2\}\}$.
In words, when both nodes are transmitting, the optimal policy in terms of total energy is to give decoding priority to the node with the worst channel~($U_{a2}$). Another consequence of Proposition~\ref{prop:SIC} is that since $\chi_{01}(t) < \chi_{02}(t), \forall t \in \mathcal{T},$ we need not specify bound \eqref{eq:CapBound1} to obtain the optimal trajectory in Figure~\ref{fig:TwoNodeSimulation}.

\begin{figure}[t!]
	\centering
	\vspace{0.5em}
	\subfloat[Transmit powers of nodes $U_{a1}$ and $U_{a2}$, and the associated decoding order $r$ at the receiving AP.]{
		\label{subfig:error1}
		\includegraphics[width=\columnwidth]{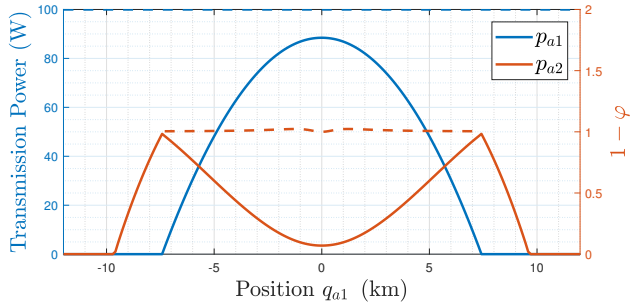}}

	\subfloat[Associated transmission rates achieved by nodes $U_{a1}$ and $U_{a2}$.]{
		\label{subfig:error2}
		\includegraphics[width=\columnwidth]{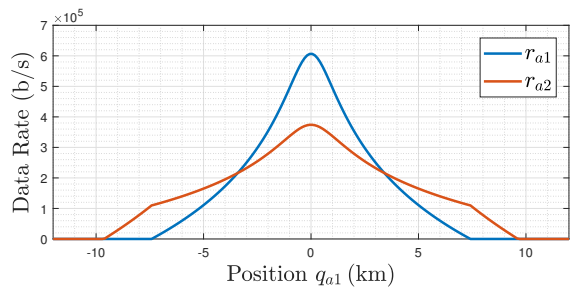} }

	\caption{Two-node transmission power problem.}
	\label{fig:TwoNodeSimulation}
\end{figure}

\section{Relay-Assisted Internet-of-Things (IoT) Network} \label{sec:RelayIoT}

The following examples are representative of a relay-assisted wireless sensor network.
We consider the geometry shown in Figure~\ref{fig:SystemModel}, where a set of terrestrial source nodes $U_{gn}, n \in \{1,\ldots,N-1\}$, are geographically isolated from AP $U_0$ and must offload their data. We first consider the relay to be an ideal sink and look at the energy savings available through joint optimization of the relay's trajectory with the source's transmissions.
We then extend these simulations by assuming the UAV $U_N$ has a finite data buffer; and hence, must relay data to $U_0$. Later in this section we will include uncertainty by considering communication over slow fading channels, and in doing so introduce successful decoding conditions.

\subsection{Open-loop Energy Savings}

We perform open-loop simulations assuming ideal AWGN channels to compare the potential energy savings.
Suppose there exist source nodes $U_{g1},U_{g2}$ that must offload all their data to a receiving UAV $U_{a1}$. The UAV operates as an ideal sink with no memory constraints, $M=D_{a1}=\infty$. Sources are initialized with a starting data load $D_{g1} = D_{g2} = 25$\,MB. Communication occurs over a two-user MAC, where the set of achievable rate tuples is upper bounded by three functions of the form \eqref{eq:CapRegion2}. This is the first result to combine mobility with transmission over a MAC.

We construct comparative schemes using the following physical network constraints. Firstly, resources may be partitioned such that there is no inter-user interference. Therefore $U_{g1},U_{g2}$ transmit on orthogonal channels of designated bandwidth\footnote{Results may be improved through optimally partitioning the bandwidth, which we do not do. Due to the identical starting loads, and similar channel gains, it is expected that equal bandwidth allocation is close to optimal.} $B_{g1}=B_{g2}=B_{a1}/2$. Partitioning $B$ is computationally simpler, since the number of constraints scale linearly (not exponentially) with the number of sources. Secondly, transmission policies may be optimized subject to a fixed UAV trajectory. In this case we assume that $U_{a1}$ moves at constant speed $v_\text{avg}$
using minimal propulsion energy. Combinations of these constraints results in four possible protocols.

Table~\ref{tab:EnergyComp} shows a comparison of the total energy usage~$\epsilon_C$, equivalent to the cost function \eqref{eq:CostFunc}, and the transmission energy $\epsilon_{g1},\epsilon_{g2}$ used by the source nodes in each scheme.
\begin{table}[tb]
\centering
\caption{Transmission energy of source $i$, ($\epsilon_{gi}, i \in \{1,2\}$) and UAV propulsion energy ($\epsilon_{a1}$) for UAV uplink under different schemes.}\label{tab:EnergyComp}
\begin{tabular}{ | c | c | c | c | c |}
 \hline
\multirow{2}{*}{}  & \multicolumn{2}{|c|}{Separate BW} & \multicolumn{2}{|c|}{Shared BW (MAC)} \\
\hline
  & $v = v_{\text{avg}}$ & $v = v^*$ & $v = v_{\text{avg}}$ & $v = v^*$ \\
 \hline
$\epsilon_{a1}$ & NA & 1 & 0.311 & 0.201 \\
$\epsilon_{g1}$ & NA & 1 & 0.381 & 0.257 \\
$\epsilon_{g2}$ & NA & 1 & 0.861 & 0.885 \\
\hline
\end{tabular}
\end{table}
All energies are given as a ratio of the worst case feasible scenario. In the simplest case, where $U_{g1},U_{g2}$ transmit over orthogonal channels and $U_{a1}$ moves at a fixed speed, the optimal transmission policy of each node is a water-filling solution, determined by a single water-filling parameter \cite{tse2005fundamentals}. This reduces the infinite-dimensional search space of the original OCP to a single dimension. However, for the given starting data load, the problem is infeasible under these conditions.

Generating a solution by solving \eqref{eq:GeneralProbForm} results in a $36\%$ total energy savings when compared  with joint optimization over single access channels, while sources $U_{g1},U_{g2}$ respectively use $80\%$ and $75\%$ less transmission energy. Although there is not significant network level energy savings for the MAC uplink under different speed regimes, both $U_{g1},U_{g2}$ save $36\%$ and $33\%$ transmission energy, respectively, by allowing the relay to vary speed. This may be of particular importance in remote sensing applications, where source nodes may have strict energy requirements or perform energy harvesting \cite{gunduz2014designing}.

\subsection{Closed-Loop Simulation}

For a closed-loop simulation, an NMPC control policy is generated by solving \eqref{eq:GeneralProbForm} at each computation interval $t_c=10$\,s, subject to initial conditions set by measured data. We consider the geometry illustrated in green in Figure~\ref{fig:SystemModel}, where sources $U_{g1},U_{g2}$ are initialised with $D_{g1}=D_{g2}=11$\,MB and $U_{a1}$ has a finite memory constraint of $M=1.5D_{g1}$. All data must be relayed to the AP $U_0$ by time $T$.
The finite time nature of the experiment motivates the use of a decreasing horizon strategy, where the final time is constant and the horizon length is reduced at each $t_c$.
The NMPC problem is solved centrally, with full state information. In practice, position and velocity information can be obtained from GPS and IMU data.

Typically data is encoded and sent in discrete codewords over packet intervals $t_p \ll t_c$. At each computation interval the complete information at each node is encoded at rates determined by \eqref{eq:GeneralProbForm}.
We assume a repeat request (ARQ) protocol. Transmitters get feedback through $1$-bit acknowledgement (ACK/NAK) signals. Buffers are only updated with successfully decoded information. Information in an unsuccessfully decoded codeword is retransmitted at a later time.

Here we assume slow fading channels, where the channel realizations are random but remain constant over~$t_p$. Therefore, at the beginning of each interval the channel state is modelled as a new realization of the random channel variable. We therefore formulate~\eqref{eq:GeneralProbForm} with $\epsilon$-outage capacity constraints, ensuring the control policy is robust to channel realizations. In the following we consider the MAC channel over a single codeword interval, dropping time dependency in notation. For actual realization $\tilde{h}_n$, channel outage --- where the codeword is not successfully decoded --- occurs because one or more of the received powers $\tilde{\beta}_n \triangleq \eta(\chi_{rn},\tilde{h}_n)p_n$ was smaller than predicted and cannot support rate $r_n$. The decoder may perform joint decoding of received signals, or decode a subset of received signals, treating others as interference. Precisely, for an $N$-user MAC, information transmitted from users in $\mathcal{S} \subseteq \mathcal{N}$ is successfully decoded if
\begin{multline} \label{eq:DecodingConditionFull}
r	\in \mathcal{D} \triangleq
\bigcup_{\mathcal{S} \subset \mathcal{N}}
\Bigg\{r>0 \mid \sum_{m \in \mathcal{M}} r_m - \\
B\log_2 \Bigg(1 + \underbrace{\frac{\sum_{n \in \mathcal{M}}\tilde{\beta}_n}{\sigma^2+\sum_{s \in \mathcal{S}'}\tilde{\beta}_s}}_\text{SNIR} \Bigg) \leq 0, \forall \mathcal{M} \subset \mathcal{S} \Bigg\}
\end{multline}
where $\mathcal{S}' \triangleq \mathcal{N}\setminus \mathcal{S}$ is the set of users not decoded, treated only as interference. Figure~\ref{fig:DecodingRegion} shows this region for an $N=2$ user MAC.
\begin{figure}[t!]
\centering
\vspace{0.5em}
\begin{tikzpicture}[scale=0.8]
	\draw[black, ultra thick, ->] (0,0) -- (0,5);
	\draw[black, ultra thick, ->] (0,0) -- (5,0);


	\draw[black, thick] (0,4) -- (2,4);
	\draw[black, thick] (4,0) -- (4,2);
	\draw[black, line width=0.4mm,] (2,4) -- (4,2);

    \draw[black,fill=black,black] (4,2) circle (.4ex);
   	\draw[black,fill=black,black] (2,4) circle (.4ex);

	\fill[pattern=horizontal lines]  (0,0) -- (5,0) -- (5,2) -- (4,2)
	 -- (2,4) -- (0,4) -- cycle;

	 \fill[pattern=north east lines]  (0,0) -- (0,5) -- (2,5) -- (2,4)
	 -- (4,2) -- (4,0) -- cycle;

	\node[draw=none, align=center] at (-0.5,4.8) {$r_{g1}$};
    \node[draw=none, align=center] at (4.8,-0.5) {$r_{g2}$};
    \node[draw=none, align=center] at (2.3,4.35) {$R^{(1)}$};
    \node[draw=none, align=center] at (4.5,2.3) {$R^{(2)}$};
    \node[draw=none, align=center] at (3.3,3.3) {$L_3$};

\end{tikzpicture}
    \caption[Two-user capacity region]{Two user decoding region. Region shaded by horizontal lines contains rates decodable from $U_{g1}$. Region shaded by vertical lines contains rates decodable from $U_{g2}$.}
    \label{fig:DecodingRegion}
\end{figure}
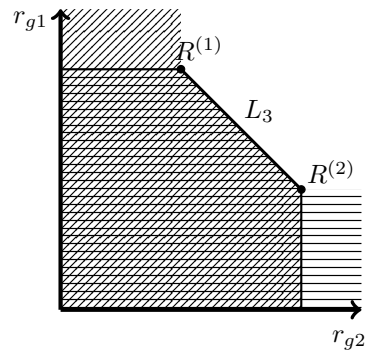
Since the rate tuple generated by \eqref{eq:GeneralProbForm} is always on the boundary of the capacity region, we may determine if $r\in\mathcal{D}$ just by considering actual and predicted received powers. If the actual SNIR is at least as large as the predicted SNIR, i.e.
\begin{equation} \label{eq:DecodingCondition}
\frac{\sum_{m \in \mathcal{M}} \beta_m}{1 + \sum_{j \in \mathcal{S'}} \beta_j} \geq \frac{\sum_{m \in \mathcal{M}} \tilde{\beta}_m}{1 + \sum_{j \in \mathcal{S'}} \tilde{\beta}_j}, \forall \mathcal{M} \subseteq \mathcal{S},
\end{equation}
where $\beta_M \triangleq \eta(\chi_{rm},F_h^{-1}(1-\epsilon))p_m$, then users $U_i,i\in\mathcal{M}$ can be decoded. 
For a single user channel, this simplifies to $\beta_{a1} \geq \tilde{\beta}_{a1}$.

UAVs are advantageous in communication networks due to line of sight (LoS) links. Multipath scattering may still occur, such as off of objects near ground nodes or flight surfaces of the UAV. Rician fading is suitable for modelling received signal strength in channels with strong LoS components \cite{zhou2012modeling}. For each channel used, $\nu$ is a vector of random variables drawn from a Rice distribution characterized by K-factor $\kappa$, defined as the ratio of received signal power in the LoS path to the power received from scattered paths. If $\kappa=\infty$ there is no fading, and the model reduces to AWGN \cite{zhou2012modeling}.
Similarly, for terrestrial applications with no LoS, $\kappa=0$ results in the commonly-used Rayleigh model \cite{zhou2012modeling}.
The cumulative distribution $\Gamma(\cdot)$ of a Rician channel is a Marcum Q-function of order 1. In our simulations we set $\kappa=10$, and assume that the fading processes of different users are independent and identically distributed,
\begin{equation}
v_{nm} \sim \text{Rice}\left(\sqrt{\kappa(\kappa + 1)^{-1}}, \sqrt{(2(\kappa+1))^{-1}}\right),
\end{equation}
for each fading instance and each pair of nodes $n,m$.

The UAV may be disturbed by wind during flight. We model wind entering the first derivative \cite{kang2009linear} such that ground speed~$\dot{q}$ is the sum of air speed $\dot{v}$ and wind speed $\dot{w}$. To account for this, the state is augmented with disturbance variable $\delta(t)$,
\begin{equation}
\dot{\delta}(t) = 0, \forall t \in \mathcal{T}, \delta(0)=w_\text{meas}
\end{equation}
and redefined position dynamics
\begin{equation}
\dot{q}(t) \triangleq \Upsilon v(t) - \delta(t).
\end{equation}
With full state information the estimate $\dot{w}_\text{meas}$ is calculated through a moving average filter. For simulation, we let$\dot{w}=-6$\,m/s.

The finite time problem may be infeasible due to channel outages and wind disturbances. In this case we switch to a variable terminal time for the last few iterations (the convergence analysis of the variable horizon scheme could be a topic of future work). Simulations are performed for $\Gamma^{-1}(1-\epsilon) \approx 0.2$, with results shown in Figures~\ref{fig:ClosedLoopSimulation1}--\ref{fig:CLosedLoopSimulation2}.
\begin{figure}[t]
	\centering
	\vspace{0.5em}
	\subfloat[$U_{a1}$ acceleration, velocity and displacement profiles.]{
		\label{subfig:1}
		\includegraphics[width=\columnwidth]{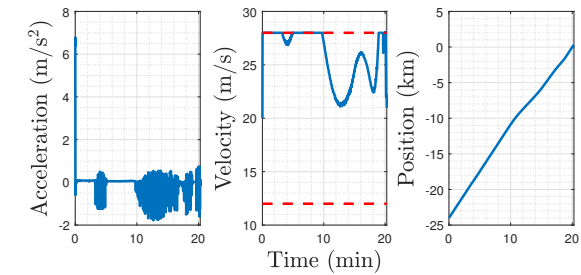} }

	\subfloat[$U_{a1}$ thrust profile.]{
		\label{subfig:2}
		\includegraphics[width=\columnwidth]{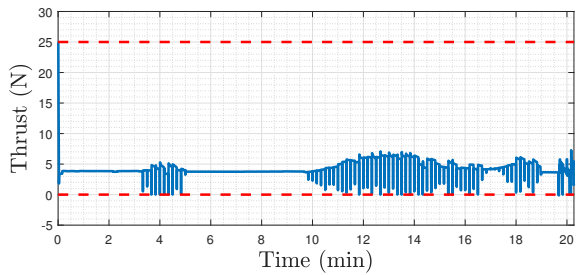} }

	\caption{Mobility related state/input trajectories (solid lines) for $U_{a1}$ during relay simulation. Dashed lines show constraints.}
	\label{fig:ClosedLoopSimulation1}
\end{figure}
\begin{figure}[t]
	\centering
	\vspace{0.5em}
	\subfloat[Data storage buffers of $U_{g1},U_{g2},U_{a1}$.]{
		\label{subfig:3}
		\includegraphics[width=\columnwidth]{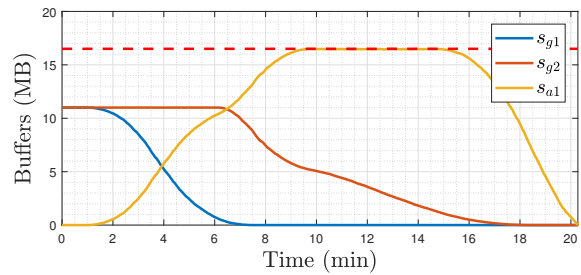} }

	\subfloat[Transmission data rates of $U_{g1},U_{g2}.U_{a1}$.]{
		\label{subfig:4}
		\includegraphics[width=\columnwidth]{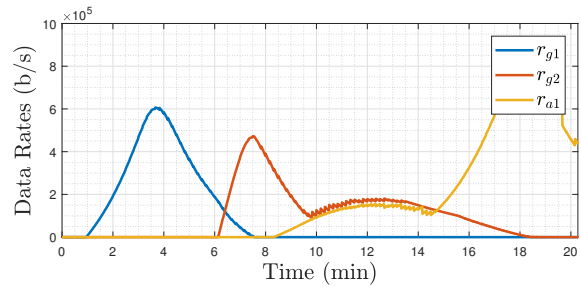} }

	\subfloat[Allocated transmission powers of $U_{g1},U_{g2},U_{a1}$.]{
		\label{subfig:5}
		\includegraphics[width=\columnwidth]{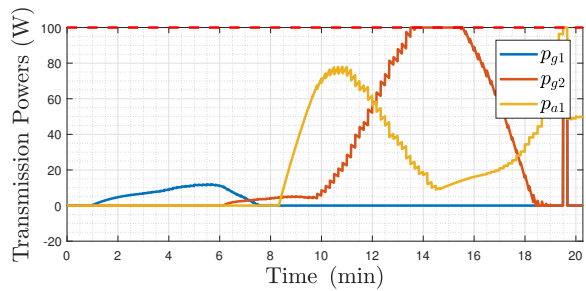}}

	\caption{Transmission related state/input trajectories (solid lines) for nodes $U_{g1},U_{g2},U_{a1}$ during UAV relay simulation. Dashed lines show hard constraints.}
	\label{fig:CLosedLoopSimulation2}
\end{figure}
In this simulation, as in Figure~\ref{fig:SystemModel}, two ground nodes $U_{g1},U_{g2}$ are relaying data to a single access point via UAV $U_{a1}$.
Figure~\ref{subfig:1} shows the mobility dynamics of the UAV, where a velocity constraint becomes active as the UAV slows down to offload collected data to the AP. The extreme change in velocity results from the UAV memory $s_{a1}$ approaching capacity. Figure~\ref{subfig:2} shows the thrust required to maintain altitude during this maneuver. Since $C_{D2} \gg C_{D1}$ in \eqref{eq:NewtDragModel}, a wind speed of $\dot{w}=-6$ beneficially slows down the UAV, reducing the minimum energy by $\approx 31\%$ compared to a wind speed of $\dot{w}=6$.

Commanded rates over $t_c$ are strict upper bounds on achievable information transfer because, even for favourable channel realizations, data will not be transferred faster than predicted.
Figures~\ref{subfig:3}--\ref{subfig:4} show data interchange between $U_i, i \in \{0,g1,g2,a1\}$ in terms of the storage memory and achieved rates.
Figure~\ref{subfig:5} shows the associated transmission power profile. Maximum power constraints are active while the UAV's buffer is close to capacity, during which the incoming and outgoing data from $U_{a1}$ are similar.  Due to the nonzero probability of outage, we cannot guarantee all data will be offloaded in $\mathcal{T}$, or indeed in any finite time. In case~\eqref{eq:GeneralProbForm} becomes infeasible, which often happens as $t \rightarrow  \sup{\mathcal{T}}$ due to the hard terminal data constraint \eqref{eq:initAndFinalStorageConst}, we allow for a variable terminal time. In practice we see that an average of $4.9$\,kB of the initial $22$\,MB remains on $s_{a1}(T)$, which takes another $1.55$\,s to offload to $U_0$, an increase in $T$ of approximately $1\%$.

\section{Conclusions and Future Work}
We have formulated nonlinear dynamic models for transmission and mobility in UAV-enabled networks.
We have considered both a Shannon capacity formulation for static AWGN channels and an outage capacity formulation for time-varying slow fading channels.
Building upon these models, we have presented a general optimization framework for joint control of propulsion and transmission in mobile communication networks.
Special cases where the OCP may either be solved to global optimality or relates to existing solutions have been discussed. In particular, when either the UAV trajectory or commanded rates across single-hop links is fixed, then the resulting constraint set and cost function are convex.
For both single- and multiple-user scenarios we have shown that significant energy savings, upwards of 70\% in some cases, are available through joint control of propulsion and transmission.

Immediate extensions of this work include higher fidelity models.
Considering a goal of on-line real-time control of multi-agent networks, the following key developments must be addressed: (i)~Closed-loop analysis of the control strategy in a decreasing or variable horizon framework, encompassing error propagation analysis; (ii)~A robust, distributed framework for the problem, to include the use of adaptive models.

All energy expenditure on an autonomous agent may be categorized as being due to propulsion, communication or computation. Considering the tangible trade-off between computation and communication energy \cite{thammawichai2018optimizing, nazemi2016qoi}, a more distant consideration is to include computation energy, such as due to data compression or aggregation \cite{orhan2015source}, into the problem.

\bibliography{main}
\bibliographystyle{ieeetr}
\end{document}